\documentclass[12pt]{article}%

\usepackage[all]{xy}

\usepackage{amsmath}
\usepackage{amsfonts}
\usepackage{amssymb}
\usepackage{graphicx}
\usepackage{pstricks}
\providecommand{\U}[1]{\protect\rule{.1in}{.1in}}

\newcommand{\N}{{\mathbb N}}
\newcommand{\C}{{\mathbb C}}

\newcommand{\R}{{\mathbb R}}

\newcommand{\Ho}{{\mathbb H}}

\newcommand{\Hh}{{\mathcal H}}
\newcommand{\Aa}{{\mathcal A}}
\newcommand{\Cc}{{\mathcal C}}

\newcommand{\Ss}{{\mathcal S}}

\newcommand{\AM}{{\mathcal{AM}}}
\newcommand{\SM}{{\mathcal{SM}}}

\def\ov{\overline}

\def\bi{\begin{itemize}}
\def\ei{\end{itemize}}
\def\un{\underline}

\newcommand{\ba}{\begin{eqnarray}}
\newcommand{\ea}{\end{eqnarray}}
\newcommand{\bas}{\begin{eqnarray*}}
\newcommand{\eas}{\end{eqnarray*}}
\newcommand{\be}{\begin{equation}}
\newcommand{\ee}{\end{equation}}

\def\dxu{\partial_{\un x}}
\def\dxv{\partial_{\vec x}}
\def\dxj{\partial_{x_j}}
\def\dx0{\partial_{x_0}}
\def\pa{\partial}

\newtheorem{theorem}{Theorem}
\newtheorem{proposition}[theorem]{Proposition}

\newtheorem{lemma}[theorem]{Lemma}

\newtheorem{definition}[theorem]{Definition}
\newenvironment{proof}[1][Proof]{\noindent\textbf{#1.} }{\ \rule{0.5em}{0.5em}}
\newtheorem{preremark}[theorem]{Remark}
\newenvironment{remark}{\begin{preremark}\rm}{\hfill$\Diamond$\end{preremark}}
\newtheorem{prenotation}[theorem]{Notation}

\numberwithin{equation}{section}
\numberwithin{theorem}{section}

\begin{document}

\title{{Extending coherent state transforms to Clifford analysis}}
\author{William D. Kirwin\thanks{Mathematics Institute, University of Cologne}, \, Jos\'e  Mour\~ao\thanks{Department of Mathematics and Center for Mathematical Analysis, Geometry and Dynamical Systems, Instituto Superior T\'ecnico, University of  Lisbon}, \, Jo\~ao P. Nunes\footnotemark[2] \, and Tao Qian\thanks{Department of Mathematics, Faculty of Science and Technology, University of Macau}}

\date{\vspace{-5ex}} 

\maketitle


\begin{abstract}
 Segal-Bargmann coherent state transforms can be viewed as unitary maps from
$L^2$ spaces of functions (or sections of an appropriate line bundle) 
on a manifold $X$ to spaces of
square integrable holomorphic functions (or sections) on $X_\C$.
It is natural to consider higher dimensional extensions of 
$X$ based on Clifford algebras as they could be useful
in studying quantum systems with internal, discrete, degrees
of freedom corresponding to nonzero spins. Notice that 
  extensions of 
$X$ based on the Grassman algebra appear naturally in the 
study of supersymmetric quantum mechanics. In Clifford analysis the zero mass 
Dirac equation provides a natural generalization of the Cauchy-Riemann conditions
of complex analysis and leads to monogenic functions. 

For the simplest but already quite interesting case of
$X = \R$ we introduce two extensions of the Segal-Bargmann coherent state transform
from $L^2(\R,dx) \otimes \R_m$ to Hilbert spaces of slice monogenic and axial monogenic functions
and study their properties. These two transforms are related by the dual Radon transform. Representation theoretic and quantum mechanical aspects of the 
new representations are studied.

\end{abstract}

\tableofcontents

\section{Introduction}

Clifford analysis (see \cite{BDS,DSS}) has been developed to extend the complex analysis of holomorphic functions to
  Clifford algebra valued functions, satisfying properties generalizing the Cauchy--Riemann conditions.

On the other hand, in quantum physics,  Clifford algebra or spinor representation valued functions 
describe some systems
with internal degrees of freedom, such as particles with spin.

Recall that the Segal-Bargmann transform \cite{Ba,Se1,Se2},
for a particle on $\R$,
establishes the unitary equivalence of  the Schr\"odinger representation
 with Hilbert space $L^2(\R, dx)$, with
(Fock space-like)  representations with Hilbert spaces,
$\Hh L^2(\C, d\nu)$,
of holomorphic functions
on the phase space, $\R^2\cong \C$ which are $L^2$ with respect to a measure $\nu$. In the Schr\"odinger 
representation the position operator $\hat x_{Sch}$ acts diagonally while the momentum operator is $\hat p_{Sch}=i\frac{d}{dx}$.
On the other hand, on the Segal--Bargmann Hilbert space 
$\Hh L^2(\C, d\nu)$ it is the operator $\hat x_{SB} + i\hat p_{SB}$ that acts as multiplication by the holomorphic function 
$x+ip$. 

In \cite{Ha1}, Hall has 
defined coherent state transforms (CSTs) for compact Lie groups $G$ which are analogs of the Segal-Bargmann transform.  
These CSTs correspond to applying heat
kernel evolution,  $e^{\frac{\Delta}2}$, followed by analytic continuation
to the complexification $G_\C$ of $G$ \cite{Ha2}.

We use the
fact that, after applying the heat kernel evolution,
the resulting functions are in fact extendable 
to $\R^{m+1}$ in two natural ways motivated by Clifford analysis.  These will lead to two
generalizations of the CST, the slice monogenic CST, $U_s$, and the axial monogenic CST, $U_a$, which take 
values on spaces of $\C_m$--valued functions on $\R^{m+1}$, where $\C_m$ denotes the
complex Clifford algebra with $m$ generators.
One, ${\mathcal H}_s= {\rm Im}\, U_s$, is a subspace of the recently introduced space of 
square integrable  slice monogenic functions \cite{CSS1}, while the other, ${\mathcal H}_a= {\rm Im}\, U_a$, is a Hilbert space of, 
the more
traditional in Clifford analysis, axial 
monogenic functions \cite{BDS, DSS}.
We show that the two coherent state transforms  are related by the 
dual Radon transform $\check R$, 
$$
U_a = \check R \circ U_s.
$$

A possibly interesting alternative way of defining a monogenic CST would be through Fueter's theorem 
\cite{F,Q,KQS,PQS,Sc}. It would be very interesting to relate such a transform with the one
studied in the present paper.

As in the case of the usual CST, the aim of these transforms is to
describe the quantum states of a particle in $\R$ with internal degrees of freedom 
parametrized by a Clifford algebra, through  slice/axial monogenic functions, thus making available,
the powerful analytic machinery  of Clifford analysis.
In Section \ref{s-4}, we show that the operator $\hat x_0 + i\hat p_0$ has a simple action in both the slice 
and axial monogenic representations.

\section{Preliminaries}
\label{s-2}

\subsection{Coherent state transforms (CST)}
\label{cst}

Let $G$ be a compact Lie group with complexification $G_\C$. 
In 1994, Hall  \cite{Ha1}  introduced a class of unitary integral transforms on $L^2(G,dx)$, 
where $dx$ is Haar measure, to spaces of holomorphic functions on $G_\C$ which are $L^2$ with 
respect to an appropriate measure. These are known as coherent state transforms (CSTs) 
or generalized Segal--Bargmann transforms. 
These transforms were extended to groups of compact type, which include the case of
$G=\R^n$ considered in the present paper, by Driver in \cite{Dr}. General Lie
groups of compact type are direct products of compact Lie groups and $\R^n$, see Corollary 2.2 of \cite{Dr}.
{}For $G=\R^n$  these transforms coincide with the classical
Segal--Bargmann transform \cite{Ba, Se1, Se2}.

We will briefly recall now
the case $G=\R$, which we will extend to the context of Clifford
analysis in the present paper. The case of arbitrary groups of
compact type is very interesting and will be studied in a forthcoming work.
Let $\rho_t(x)$  denote the fundamental solution of the
heat equation.
$$
\frac{\partial}{\partial t} \rho_t =  \frac 12 \Delta \, \rho_t ,
$$
i.e.
$$
\rho_t(x) = \frac{1}{(2\pi t)^{1/2}} \, e^{-\frac{x^2}{2t}},
$$
where $\Delta$ is the Laplacian for the Euclidian metric and
let $\Hh(\C)$ denote the space of entire holomorphic
functions on $\C$.
The Segal--Bargman or coherent state transform
$$
U \ : L^2(\R, dx) \longrightarrow \Hh(\C)
$$
is defined by
\ba
\label{ee-cst} \nonumber
U(f) (z) &=& \int_{\R} \, \rho_{1}(z - x) f(x) \, dx = \\
&=& \frac{1}{(2\pi )^{1/2}} \, \int_{\R}  \, e^{-\frac{(z-x)^2}{2}} \, f(x)
\, dx \, .
\ea
where $\rho_1$ has been analytically continued to $\C$.
We see that the transform $U$ in (\ref{ee-cst}) factorizes according
to the following diagram

\vskip 0.3cm

\begin{align}
 \label{d1}
\begin{gathered}
\xymatrix{
&&{\mathcal H} (\C)  \\
L^2(\R, dx)  \ar@{^{(}->}[rr]_{e^{\frac{\Delta}2}}  \ar@{^{(}->}[rru]^{U} && \Aa (\R)
\ar[u]_{\Cc}
  }
\end{gathered}
\end{align}

\noindent 
where $\Aa(\R)$ is the space of (complex valued) 
real analytic functions on $\R$ with unique analytic continuation to entire functions  on $\C$, $\Cc$ denotes the  analytic continuation from $\R$ to $\C$
and $e^{\frac{\Delta}2}(f)$ is the (real analytic) heat kernel evolution of the function $f\in L^2(\R,dx)$ at time $t=1$,
that is the solution of
\be
\label{ee-he}
\left\{
\begin{array}{lll}
\frac{\partial}{\partial t} \, h_t &=& \frac 12 \Delta \, h_t\\
h_0 &=& f
\end{array}   ,
\right.
\ee
evaluated at time $t=1$,
$$
e^{\frac{\Delta}2}(f) = h_1.
$$
Let $\widetilde \Aa (\R)\subset \Aa(\R)$  denote the image of $L^2(\R,dx)$  by the operator
$e^{\frac{\Delta}2}$. 

\begin{theorem} [Segal--Bargmann]
\label{th-cst}
The transform  (\ref{ee-cst})
\begin{align}
 \label{d2}
\begin{gathered}
\xymatrix{
&&  \Hh L^2(\C, \nu \,  dx dy)   \\
L^2(\R, dx)  \ar[rr]_{e^{\frac{\Delta}2}}  \ar[rru]^{U} &&
\widetilde \Aa (\R)
\ar[u]_{\Cc}
  }
\end{gathered}
\end{align}
 is a unitary isomorphism,
where $z = x + i y \in \C, x, y \in \R$ and
$\nu(y) = e^{-y^2}$.
\end{theorem}

\subsection{Clifford analysis}
\label{ss-ca}

Clifford analysis has been developed to extend the complex analysis of holomorphic
functions 
  to
  Clifford algebra valued functions, satisfying properties generalizing the Cauchy--Riemann conditions \cite{BDS, DSS}. Let us briefly recall from \cite{CSS1} and \cite{DS},  some definitions and results from  Clifford analysis.
 Let $\R_m$ denote the real Clifford algebra with $m$ generators, $e_j, j = 1, \dots, m$, identified
with the canonical basis of $\R^m \subset \R_m$ and satisfying the relations $e_i e_j + e_je_i = - 2 \delta_{ij}$.
We have that $\R_m = \oplus_{k=0}^m \R_m^k$,
where $ \R_m^k$ denotes the space of $k$-vectors, defined by $\R_m^0 = \R$ and  $\R_m^k = {\rm span}_\R \{e_A \ : \ A \subset \{1, \dots , m\}, |A| = k\}$.  We see that, in particular, $\R^m$ is
identified with the space of $1$-vectors, $\R^m = \R_m^1$, $\un x = \sum_{j=1}^m x_j e_j$
and $\R^{m+1}$ is identified with the space, $\R_m^{\leq 1}$,  of paravectors of the form, $$\vec x = x_0 + \un x =
x_0 + \sum_{j=1}^m \, x_j e_j.$$
Notice also that $\R_1 \cong \C$ and $\R_2 \cong \Ho$.
The inner product in $\R_m$ is defined by
$$
<u, v> = <\sum_A u_A e_A ,  \sum_B v_B e_B> = \sum_A u_A v_A ,
$$
and therefore, $\un x^2 = - |\un x|^2 = - <x, x>.$
The Dirac operator is defined as
$$
\dxu = \sum_{j=1}^m \, \dxj e_j,
$$
and the Cauchy-Riemann operator as
\be
\nonumber
\dxv = \dx0 + \dxu .
\ee
We have that $\dxu^2 = - \sum_{j=1}^m  \frac{\partial^2}{\partial x_j^2}$ and
$\dxv \ov \dxv  = \sum_{j=0}^m  \frac{\partial^2}{\partial x_j^2}$.

Recall that a continuously  differentiable function $f$ on an open domain $U\subset \R^{m+1}$,  
with values on $\R_m$ or $\C_m= \R_m \otimes \C$, 
is called (left) monogenic on $U$  if (see, for example, \cite{BDS,DSS})
\be\nonumber
\dxv f(x_0, \un x)  = (\dx0 + \dxu) f(x_0, \un x) = 0.
\ee
For $m=1$, monogenic functions on $\R^2$ correspond to holomorphic
functions of the complex variable $x_0+e_1 x_1$.

\section{Monogenic extensions of analytic functions}

\subsection{Slice monogenic extension}
\label{ss-sme}

Recall from \cite{CSS1,CSS2} that   a function $f \ : \ U \subseteq \R^{m+1} \rightarrow \R_m$ 
is slice
monogenic if, for any unit vector $\un \omega \in S^{m-1}=\{\un x \in \R_m^1 \, : \, 
|\un x| = 1\}$, the restrictions $f_{\un \omega}$  of $f$ to the complex planes
$$
H_{\un \omega} = \left\{u+v \, \un \omega, \,\, u,v \in \R\right\}, 
$$ are holomorphic,
\be
\label{ee-sm}
(\pa_u + \un \omega \, \pa_v)f_{\un \omega} (u,  v) = 0, \, \,   \forall
\un \omega \in S^{m-1}.
\ee
Let $\SM(\R^{m+1})$
denote the space of slice monogenic functions on $\R^{m+1}$.
From the definition of ${\mathcal A}(\R)$ in diagram (\ref{d1}) and the Remark 3.4 of \cite{CSS3} 
(see also Proposition 2.7 in \cite{CSS2h}) one obtains the following.
\begin{theorem} 
\label{th-scke}
The slice-monogenic extension map,
\ba
\label{ee-sme}
\nonumber M_s \ : \ \Aa(\R)\otimes \R_m & \longrightarrow &   \SM(\R^{m+1})  \\
\nonumber M_s(h)(x_0, \un x) &=& M_s(\sum_A h_A \, e_A)(x_0, \un x) = \\
 &=& \sum_A h_A(x_0 + \un x) \, e_A := \sum_A e^{\un x \, \frac{d}{dx_0}} \, h_A(x_0) \, e_A =\\
&=& \nonumber \sum_A \sum_{k=0}^\infty \frac {\un x^k}{k!}  \, \frac{d^k h_A}{dx_0^k}(x_0) \, e_A  ,
\ea
is well defined and satisfies $M_s(h)(x_0, 0) = h(x_0), \forall x_0 \in \R$.
\end{theorem}

\subsection{Axial monogenic extension and dual Radon transform}
\label{ss-ame}

A monogenic  function $f(x_0, \un x)$ is called axial monogenic (see \cite{DS}, p. 322,  
for the definition of axial monogenic functions of degree $k$) if it is
of the form
\ba
\label{ee-axm}
\nonumber f(x_0, \un x) &=&  \sum_A \, f_A(x_0, \un x) \, e_A \\
f_A(x_0, \un x) &=& B_A(x_0, |\un x|) + \frac{\un x}{|\un x|} \, C_A(x_0, |\un x|) \,,
\ea
where $B_A, C_A$ are scalar functions and the 
functions  $f_A$ are monogenic.
The monogeneicity condition, $\dxv f_A = \dx0 f_A + \dxu f_A =0$, then
leads to the Vekua--type system for $B_A, C_A$, generalising the
Cauchy-Riemann conditions,
$$
\dx0 B_A - \partial_r C_A = \frac{m-1}r \, C_A \, , \quad \dx0 C_A + \partial_r B_A = 0, \, r = |\un x|.
$$
 Let $\AM (\R^{m+1})$
denote the space of axial monogenic functions on $\R^{m+1}$.

Axial monogenic functions are determined by their restriction to the real axis,  $f(x_0, 0)$. 
The inverse map of extending (when such an extension exists)
a real analytic function  $h$  on $\R$
 to an axial monogenic function on $\R^{m+1}$
 is called generalized axial Cauchy-Kowalewski extension and 
has been studied by many authors (see, for example, \cite{DS}). 

Using the dual Radon transform to map slice monogenic 
functions to monogenic
functions as proposed in \cite{CLSS},
we will factorize the axial monogenic extension into the slice monogenic
extension followed by the dual Radon transform.
Let us first recall the definition of the 
dual Radon transform. (See, for example, \cite{He}.)

\begin{definition}
\label{de-drt}
The dual Radon transform of a smooth function $f$ on $\R^{m+1}$ is
\be
\label{ee-drt}
\check R ( f)(x_0, \un x) = \int_{S^{m-1}} \, f(x_0, < \un x, \un t> \un t) \, d \un t .
\ee
\end{definition}
It is known from \cite{CLSS} that $\check R$ maps entire slice monogenic functions
to entire monogenic functions.

Let us denote
a function $f\in \Aa(\R)$ and its analytic continuation to the complex plane $H_{\un t}$ by the same symbol, $f$. The following is a small modification of the Theorem 4.2 in \cite{DS}.

\begin{theorem}
\label{th-cke}
The axial monogenic or axial Cauchy-Kovalewski extension map
\ba
\label{ee-cke}
\nonumber M_a \ : \ \Aa(\R)\otimes \R_m & \longrightarrow &   \AM(\R^{m+1})  \\ \nonumber 
M_a(h)(x_0, \un x) &=&  M_a(\sum_A h_A e_A)(x_0, \un x) = \\
&=& 
\sum_A \, 
\int_{S^{m-1}} h_A (x_0 \, +  \, <\un x, \un t>\un t) \, d \un t    \, e_A ,
\ea
where $d \un t$ denotes the invariant normalized (probability) measure on $S^{m-1}$,
is well defined and satisfies $M_a(h)(x_0, 0) = h(x_0), \forall x_0 \in \R=\R_m^0$.
\end{theorem}

\begin{proof}
{}From (\ref{ee-sme}) and (\ref{ee-drt}) we see that the map $M_a$
in (\ref{ee-cke}) factorizes to
\be
\label{ee-ame2}
M_a = \check R \circ M_s \ .
\ee
The fact that the image of this map is a subspace of
the space of entire monogenic functions on $\R^{m+1}$ is
a consequence of the theorem A of \cite{CLSS}. 
We still need
to show that the functions $M_a(h)$ are 
axial monogenic for all $h\in \Aa(\R)\otimes \R_m$. Notice that
the Taylor series of $h$, with center at any point of $\R$ has
infinite radius of convergence.
Using (\ref{ee-sme}), Theorem \ref{th-scke}, and the fact 
that for $\un\omega\in S^{m-1}$ one has $\un\omega^{2k}=(-1)^k$, we obtain
\begin{align}
\nonumber M_a(h)(x_0,\un x) &= M_a(\sum_A h_A \, e_A)(x_0,\un x) = \\
=\sum_A \check R\circ M_s(h_A)(x_0,\un x) \, e_A &= \sum_A \int_{S^{m-1}}\sum_{k=0}^\infty \frac {(\langle \un x,\un\omega\rangle\un\omega)^k}{k!}  h_A^{(k)}(x_0)d\un\omega \, e_A\nonumber \\  
 \nonumber 
= \sum_A \left(\sum_{j=0}^\infty\int_{S^{m-1}} \frac{(-1)^j}{(2j)!}h_A^{(2j)}(x_0) \langle \un x,\un
\omega\rangle^{2j} \right. & \left. + \un \omega\,\frac{ (-1)^j}{(2j+1)!} \,  \, h_A^{(2j+1)}(x_0) \langle \un x,\un\omega\rangle^{2j+1} d\un\omega\right) \, e_A . 
\end{align}
and therefore,
$$
M_a(h)(x_0,\un x) = \sum_A\left(\sum_{j=0}^\infty \frac{(-1)^j}{(2j)!}h_A^{(2j)}(x_0)C_{m,2j}|\un x|^{2j} + \un x\,\frac{ (-1)^j}{(2j+1)!}h_A^{(2j+1)}(x_0) C_{m,2j+2}|\un x|^{2j }\right)\, e_A,
$$
where
$$
C_{m,2j}=\int_0^{\pi} \sin^{m-1}(\theta) \cos^{2j}(\theta)\,d\theta.
$$
This is of the form (\ref{ee-axm}) which completes the proof.
\end{proof}

We therefore get the following commutative diagram.
\begin{align}
 \label{d22}
\begin{gathered}
\xymatrix{
&&  \AM (\R^{m+1})  \\
\mathcal{A} (\R) \otimes \R_m  \ar[rrd]_{M_s}  \ar[rru]^{M_a} && \\
&& \SM (\R^{m+1}) \ar[uu]_{\check R}
  }
\end{gathered}
\end{align}

As an illustration let us consider the axial monogenic extension of plane waves  $\varphi_p$, with $\varphi_p(x_0)=e^{ipx_0}$.
The axial monogenic extension of $\varphi_p$ follows from Example 2.2.1 and Remark 2.1 of \cite{DS}, 
where the axial monogenic extension of $e^{x_0}$ is given in terms of Bessel functions, by taking $k=0$ and replacing 
$\un x$ by $ip \un x$ in the expressions of Example 2.2.1 of \cite{DS}.
\begin{proposition}
\label{le-ampw}
 The axial monogenic plane waves are given by
\be
\label{ee-ampw}
M_a(\varphi_p)(x_0, \un x) = 
\, \Gamma(\frac m2) \left(\frac{2i}{p|\un x|}\right)^{m/2-1}\left( I_{m/2-1}(p |\un x|) + i \frac{\un x}{|\un x|} \,  I_{m/2}(p |\un x|) \right) \, e^{ipx_0} \, ,
\ee
where $I_\alpha$ are the hyperbolic Bessel functions.
\end{proposition}

\begin{proof}
By representing, as in example 2.2.1 of \cite{DS}, $M_a(\varphi_p)(x_0) $ in the form
$$
M_a(\varphi_p)(x_0, \un x) = \sum_{j=0}^\infty \, c_j \un x^j B_j e^{ipx_0},
$$
and expressing the monogeneicity of the transform
$$
(\partial_{x_0} + \partial_{\un x}) \,  \sum_{j=0}^\infty \, c_j \un x^j B_j e^{ipx_0} = 0,
$$
we obtain the following recurrence relation for the functions $B_j(x_0)$,
$$
B_{j+1}(x_0) = - ip B_j(x_0) - B'_j(x_0),  \quad   B_0(x_0)=1.
$$
The solution  is $B_j(x_0)= (-ip)^j$. Then we see that the transform is obtained by replacing
$\un x$ by $ip \un x$ in the expressions of Example 2.2.1 of \cite{DS}.
\end{proof}

\begin{remark}\label{injective}From Theorem A of \cite{CLSS}, 
$\check R: \SM (\R^{m+1})\to \AM (\R^{m+1})$ is an injective map. In fact, 
from Corollary 4.4 of \cite{CLSS}, we see that 
(non-zero) slice monogenic functions do not belong to ${\rm Ker}\, \check R$.
\end{remark}


\begin{remark}Note that, as in \cite{DS}, considering $h\in {\mathcal A}(\R)\otimes \C_m$, one also has, 
\be\label{axialtoo}
M_a(h)(x_0,\un x) = \sum_A\int_{S^{m-1}} h_A (x_0+i\langle \un x, \un t\rangle )(1-i\un t)) \, d\un t\, e_A,
\ee
which is equivalent to (\ref{ee-cke}) and can be readily verified by expansion in power series.
\end{remark}

\section{Clifford extensions of the CST}
\label{s-3}

The two extensions  (\ref{ee-sme}) and  (\ref{ee-cke}) give two natural
paths to define coherent state transforms by replacing the
vertical arrow of analytic continuation in the diagram
(\ref{d2}).

We refer the reader interested in the representation theoretic and the
quantum mechanical meaning  of the Hilbert spaces introduced in 
the present section to section \ref{s-4}.

\subsection{Slice monogenic coherent state transform (SCST)}
\label{ss-31}

The slice monogenic CST is naturally defined by substituting the vertical arrow
in the diagram (\ref{d2}) by the slice monogenic extension
(\ref{ee-sme}) leading to

\begin{align}
 \label{d3}
\begin{gathered}
\xymatrix{
&&  \SM (\R^{m+1})\otimes \C  \\
L^2(\R, dx_0)\otimes \C_m  \ar@{^{(}->}[rr]_{e^{\frac{\Delta_0}2}}  \ar@{^{(}->}[rru]^{U_s} && \widetilde \Aa (\R)\otimes \C_m
\ar[u]_{M_s}
  }
\end{gathered}
\end{align}

\noindent where $\Delta_0 = \frac{d^2}{dx_0^2}$.
Notice that even though the plane waves, $\varphi_p(x_0)=e^{ipx_0}$, are not in the
Hilbert space $L^2(\R, dx_0)$, they are
generalized eigenfunctions
of $\Delta_0$ with  eigenvalue $-p^2$, and therefore
\be
\label{ee-pw}
e^{\frac{\Delta_0}2}(\varphi_p)(x_0) =
e^{\frac{\Delta_0}2} \, e^{ipx_0} = e^{- \frac{p^2}2} \,  e^{ipx_0} = e^{- \frac{p^2}2} \varphi_p(x_0).
\ee
On the other hand since the plane waves $\varphi_p \in \Aa(\R)$ we can
use (\ref{ee-sme}) to obtain the following  very simple result.
\begin{lemma}
\label{le-smpw} The slice monogenic plane waves are given by
\be
\label{ee-smpw}
M_s(\varphi_p)(x_0) = M_s(e^{ipx_0}) = e^{i p \vec x} =
\left(\cosh(p |\un x|) + i \frac{\un x}{|\un x|} \, \sinh(p |\un x|) \right)\, e^{ipx_0} .
\ee
\end{lemma}
\begin{proof}
{}From (\ref{ee-sme}) we obtain
$$
M_s(\varphi_p)(x_0)  = e^{ipx_0} \, \sum_{k=0}^\infty \, \frac{ (ip\un x)^k}{k!}
= \left(\cosh(p |\un x|) + i \frac{\un x}{|\un x|} \, \sinh(p |\un x|) \right)\, e^{ipx_0} .
$$

\end{proof}

\begin{proposition}
\label{ee-scst}
Let $f \in L^2(\R, dx_0)$ and
$$
f(x_0)= \frac 1{\sqrt{2\pi}} \, \int_\R \, e^{ipx_0} \, \tilde f(p) \, dp.
$$
We have
\ba
\label{ecthar}
U_s(f)(x_0, \un x)   &=&  \frac 1{\sqrt{2\pi}} \int_\R \, e^{- \frac {p^2}2} \, e^{ip\vec x} \,  \tilde f(p) \, dp =  \\
 \nonumber &=& \frac 1{\sqrt{2\pi}} \int_\R \, e^{- \frac {p^2}2} \, e^{ip x_0} \,
\cosh(p |\un x|) \,  \tilde f(p) \, dp  + i \frac{\un x}{|\un x|} \frac 1{\sqrt{2\pi}} \int_\R \, e^{- \frac {p^2}2} \, e^{ipx_0} \,  \sinh(p |\un x|) \,\tilde f(p) \, dp
\ea
\end{proposition}
\begin{proof}
This result follows from  Lemma \ref{le-smpw}, (\ref{ee-sme}) and (\ref{ee-pw}).
\end{proof}

Consider the standard inner product on $\C_m$.
Our main result in this section is the following.
\begin{theorem}
\label{th-scst}
The SCST, $U_s$   in Diagram (\ref{d3}),  is unitary onto its image  for the measure
 $d\nu_m$
on $\R^{m+1}$  given by
$$
d\nu_m  =  {\frac{2}{\sqrt{\pi} }} \,  \frac 1{Vol(S^{m-1})}  \, \frac{e^{-  |\un x|^2}}{|\un x|^{m-1}} \, dx_0 d\un x ,
$$
where $Vol(S^{m-1})$ denotes the volume of the unit radius sphere in $\R^m$, i.e. the map
$U_s$ in the diagram
\begin{align}
 \label{d4}
\begin{gathered}
\xymatrix{
&&  {\mathcal H}_s  \\
L^2(\R, dx_0)\otimes \C_m  \ar[rr]_{e^{\frac{\Delta_0}2}}  \ar[rru]^{U_s} && \widetilde \Aa (\R) \otimes \C_m
\ar[u]_{M_s}
  }
\end{gathered}
\end{align}
is a unitary isomorphism, where ${\mathcal H}_s =U_s(L^2(\R, dx_0)\otimes \C_m) \subset\SM L^2 (\R^{m+1}, 
d\nu_m)$.

\end{theorem}

\begin{proof}
Let $\Ss(\R)$ be the space of Schwarz functions on $\R.$
For $f,h \in \Ss(\R) \otimes \R_m$, with $f=\sum_A f_A e_A, h=\sum_Ah_A e_A$ we have
\bas
<U_s(f), U_s(h)> &=&
\frac 2{\sqrt \pi}
\frac{1}{Vol(S^{m-1})} \sum_A \, \int_{\R  \times \R^m} \,  \left[e^{ 2i \un x p}  \right]_0
\, e^{-{p^2}} \, \tilde f_A(p) \ov{\tilde h_A}(p) \,
\frac{e^{-|\un x|^2}}{|\un x|^{m-1}} d^mx dp = \\
&=& \frac 2{\sqrt \pi}
\frac{1}{Vol(S^{m-1})} \sum_A \,\int_{\R} \,
 e^{-{p^2}} \, \tilde f_A(p) \ov{\tilde h_A}(p) \,
\left(\int_{\R^m} \cosh({2  |\un x| p})  \frac{e^{-|\un x|^2}}{|\un x|^{m-1}}\, d^mx\right)  dp = \\
&=&
\frac 2{\sqrt \pi}
\frac{1}{Vol(S^{m-1})}\sum_A \,\int_{\R} \,
 e^{-{p^2}} \, \tilde f_A(p) \ov{\tilde h_A}(p) \,
\left(\int_{0}^\infty \cosh({2  u p})  {e^{- u^2}}\, du \right) dp = \\
&=&   \sum_A \,\int_{\R} \,
  \, \tilde f_A(p) \ov{\tilde h_A}(p) \, dp = <f, h>.
\eas
{}From the density of $\Ss(\R)  \otimes \C$ in $L^2(\R)$ we conclude that $U_s$ is unitary onto its image.
\end{proof}

\begin{remark}
\label{rm-cxpl}
For each complex plane $H_{\un\omega}:=\{u+v\un\omega:u,v\in\mathbb{R}\}$ and for 
$f\in L^2(\R,dx)\otimes \C_m$, $f=\sum_A f_A\,e_A$, 
the map $f\mapsto U_s(f)\big\vert_{H_{\un\omega}}$ coincides, 
for each component $f_A$ of $f$, 
with the Segal--Bargmann transform, which is surjective to 
${\mathcal H}L^2(H_{\un\omega}, d\nu)$ and unitary for 
the measure $d\nu=e^{-v^2}du\,dv$ on $H_{\un\omega}$.
\end{remark}

\subsection{Axial monogenic coherent state transform (ACST)}
\label{ss-32}

The axial monogenic CST is  also naturally defined as the heat kernel evolution followed 
by the axial Cauchy-Kowalewski extension
$$
U_a  = M_a \circ e^{\frac{\Delta_0}2} \, ,  
$$
i.e. by substituting the vertical arrow
in the diagram (\ref{d1}) by the axial monogenic extension (\ref{ee-cke})

\begin{align}
 \label{d5}
\begin{gathered}
\xymatrix{
&&  \AM (\R^{m+1}) \otimes \C \\
L^2(\R, dx_0) \otimes \C_m \ar@{^{(}->}[rr]_{e^{\frac{\Delta_0}2}}  \ar@{^{(}->}[rru]^{U_a} && \Aa (\R) \otimes \C_m
\ar[u]_{M_a}
  }
\end{gathered}
\end{align}

The following is an easy consequence of Theorem \ref{th-scst}, (\ref{ee-ame2}) and Remark \ref{injective}.
\begin{theorem}
Let ${\mathcal H}_a\subset  \AM (\R^{m+1}) \otimes \C$ denote the image of $L^2(\R,dx_0)\otimes \C_m$ under $U_a$.
The restriction of the dual Radon transform to ${\mathcal H}_s$  defines an isomorphism 
to   ${\mathcal H}_a$.

The diagram 
\begin{align}
 \label{d7}
\begin{gathered}
\xymatrix{
&&&&  {\mathcal H}_a  \\
L^2(\R, dx_0)\otimes \C_m  \ar[rr]^{\quad e^{{\Delta_0}/2}}  \ar[rrrru]^{U_a} \ar[rrrrd]_{U_s} && \widetilde \Aa (\R)
\otimes \C_m
\ar[rru]_{M_a} \ar[rrd]^{M_s}  \\
&&&&  {\mathcal H}_s.  \ar[uu]_{\check R}
  }
\end{gathered}
\end{align}
is  commutative and its exterior arrows are unitary isomorphisms for the inner product
on ${\mathcal H}_a$ given by
$\langle\cdot,\cdot\rangle_{{\mathcal H}_a}$,
\be
\label{ee-ip}
\langle F,G\rangle_{{\mathcal H}_a}=\int_{\R^{m+1}} (\check R)^{-1}(F)(\check R)^{-1}(G)d\nu_m,
\ee
where $d\nu_m$ was defined in Theorem \ref{th-scst}.
\end{theorem}

\begin{proof}
The injectivity of $\check R_{|{\mathcal H}_s}$ follows from Remark \ref{injective}. From (\ref{ee-ame2}), 
we conclude that $U_a = \check R \circ U_s$ which implies the surjectivity of 
$\check R_{|{\mathcal H}_s} : {\mathcal H}_s \longrightarrow {\mathcal H}_a$. Then, the 
inner product (\ref{ee-ip}) is well defined, the diagram (\ref{d7}) is commutative and
the exterior arrows are unitary isomorphisms.
\end{proof}

\begin{remark}
As mentioned in the introduction, a possibly interesting alternative way of defining a 
monogenic CST would be by replacing the dual Radon transform in 
(\ref{d22}) and in diagram (\ref{d7}) by the Fueter mapping, $\Delta^{\frac{m-1}2}$,
where $\Delta = \sum_{j=0}^m  \frac{\partial^2}{\partial x_j^2}$ (see \cite{F,Q,KQS,PQS,Sc}). 
Notice however that the map $\Delta^{\frac{m-1}2} \circ M_s$
does not correspond to a monogenic extension of analytic functions of one variable
as the restriction to the real line does not give back the original functions.
It leads nevertheless to an interesting transform and it would be very interesting to relate it with 
$U_a$.

\end{remark}

\section{Representation theoretic and quantum mechanical interpretation}
\label{s-4}

Recall that the Schr\"odinger representation in quantum mechanics is the representation for which the position 
operator $\hat x_0$ acts by multiplication on $L^2(\R,dx_0)$. The momentum operator is then given by 
$$
\hat p_0 = i \frac{d}{dx_0}.
$$

The CST from Section \ref{cst} intertwines the Schr\"odinger 
representation with the Segal-Bargmann representation, on which
the operator $\hat x_0+i\hat p_0$ acts as the operator of multiplication 
by the  holomorphic function $x_0+ip_0$ (see Theorem 6.3 of \cite{Ha2})
\be\label{cstacts}
\left( U\circ (\hat x_0 +i\hat p_0) \circ U^{-1}\right)(f)(x_0,p_0)=(x_0+ip_0)f(x_0,p_0).
\ee

We will prove now the analogous result that
the slice monogenic CST intertwines the Schr\"odinger representation with the
 representation on which  
$\hat x_0+i \hat p_0$ acts as the operator of left multiplication 
by the slice monogenic function $x_0+ \un x$.
\begin{proposition}
\label{th-scts} 
The observable $x_0+ip_0$ is represented in the slice monogenic representation
by the operator of multiplication by the slice monogenic function $x_0+\un x$, i.e.
\be\label{scstacts}
\left (U_s \circ (\hat x_0 + i \hat p_0) \circ U_s^{-1}\right)(f)(x_0,\un x) = (x_0+\un x) f(x_0,\un x),\,\, f\in {\mathcal H}_s.
\ee
\end{proposition}

\begin{proof}
We have $U_s = M_s \circ e^{\frac{\Delta_0}{2}}$. From the injectivity of the slice monogenic extension map $M_s$, 
(\ref{scstacts}) is equivalent to
$$
\left(e^{\frac{\Delta_0}{2}} \circ (x_0-\frac{d}{dx_0})\circ e^{-\frac{\Delta_0}{2}}\right) (f) (x_0)= x_0f(x_0).
$$
This follows from Theorem 6.3 of \cite{Ha2}.
\end{proof}

For the axial monogenic coherent state transform defining the axial monogenic representation, on the other hand, we have a more complicated representation of $x_0+ip_0$
involving the Cauchy-Kowalesky 
extension of the polynomials $\un x^j, j\in \N_0$.

Recall, from Theorem 2.2.1 of \cite{DSS}, that the Cauchy-Kowalesky extension of $\un x^j$ is given by the polynomial
$X^{(j)}_0(x_0,\un x)$, such that $X^{(j)}_0(0,\un x)= \un x^j$, where
$$
X^{(j)}_0(x_0,\un x)= CK(\un x^j) = \mu_0^j |x|^j \left(C_j^{(m-1)/2}\left(\frac{x_0}{|x|}\right) + \frac{m-1}{m+j-1} C_{j-1}^{(m+1)/2}
\left(\frac{x_0}{|x|}\right)\frac{\un x}{|x|}\right),
$$
with 
$$
\mu_0^{2j} = (-1)^j (C_{2j}^{(m-1)/2}(0))^{-1}, \, \, \mu_0^{2j+1}=(-1)^j \frac{m+2j}{m-1} (C_{2j}^{(m+1)/2}(0))^{-1}
$$
and the Gegenbauer polynomials  
$$
C_j^\nu(y) = \sum_{i=0}^{[j/2]} \frac{(-1)^i (\nu)_{j-i}}{i! (j-2i)!}(2y)^{j-2i},
$$
where $(\nu)_j=\nu (\nu+1) \cdots (\nu + j-1)$.

\begin{proposition} 
\label{th-42} Let $f\in {\mathcal H}_a$ be given by
\be\label{expansion}
f(x_0,\un x) = \sum_{i=0}^\infty X_0^{(i)}(x_0,\un x)\, f_{i} .  
\ee
The observable $x_0+ip_0$ is represented in the axial monogenic representation
by the following operator
\be
\label{ee-th41}
\left( U_a \circ (\hat x_0 + i \hat p_0) \circ U_a^{-1}\right) \left(f\right)(x_0,\un x) 
=  \sum_{i=0}^\infty \left(
\frac{2i+1}{2i+m} X_0^{(2i+1)}(x_0,\un x)f_{2i} + 
X_0^{(2i+2)}(x_0,\un x)f_{2i+1} \right)  
. 
\ee
\end{proposition}

\begin{proof}
From Theorem 3.4 of \cite{CLSS}, any entire axial monogenic function has an expansion of the form 
(\ref{expansion}). On the other hand, from equations (22) and (23) of \cite{CLSS} we obtain
\bas
\left( \check R \circ (x_0 + \un x) \circ {\check R}^{-1}\right)\left(X_0^{2j}\right) &=& \frac{2j+1}{2j+m} \, X_0^{2j+1}
\\
 \,\left( \check R \circ (x_0 + \un x) \circ {\check R}^{-1}\right)\left(X_0^{2j+1}\right) &=&  \, X_0^{2j+2} \, , \qquad   j \in \N_0 .
\eas
These identities, together with the Proposition \ref{th-scts} and the fact that $U_a=\check R \circ U_s$
prove (\ref{ee-th41}).
\end{proof} 

\begin{remark}
On the axial monogenic representation, one does not expect to have operators of multiplication by nontrivial functions
as the product of monogenic functions is in general not monogenic. The axial monogenic 
representation of $x_0+ip_0$ given by (\ref{ee-th41}) is in a sense the closest one can get to such an operator 
as it maps the monogenic polynomial of order $k$,  $X_0^k = CK(\un x^k)$, 
to a scalar times the monogenic polynomial of order $k+1$, $X_0^{k+1}$.

\end{remark}

\bigskip
\bigskip

{\bf \large{Acknowledgements:}} The authors would like to thank the referee for several suggestions 
and corrections. The authors were partially
supported by Macau Government FDCT through the project 099/2014/A2,
{\it Two related topics in Clifford analysis}. The authors  
JM and JPN were also partly supported by
FCT/Portugal through the projects UID/MAT/04459/2013, 
EXCL/MAT-GEO/0222/2012, PTDC/MAT-GEO/3319/2014. JM was also partially supported by the 
Emerging Field Project on Quantum Geometry from Erlangen--N\"urnberg University.

\end{document}